\documentclass[sigconf]{acmart}

\usepackage{amssymb, amsmath}
\usepackage{graphicx}
\usepackage{tikz}
\usepackage{standalone}
\usepackage{url}
\usepackage{wrapfig}
\usepackage{array}
\usepackage[inline]{enumitem}

\usepackage{footnote}

\makesavenoteenv{table}

\newcommand{\ignore}[1]{}

\usepackage{standalone}

\usetikzlibrary{calc,trees,positioning,arrows,chains,shapes.geometric,%
  decorations.pathreplacing,decorations.pathmorphing,shapes,%
  decorations.text, fit, shapes.geometric, hobby, backgrounds, calc,
  matrix,shapes.symbols,plotmarks,decorations.markings,shadows}

\makeatletter
\DeclareRobustCommand*\cal{\@fontswitch\relax\mathcal}
\makeatother

\setcopyright{none}
\renewcommand\footnotetextcopyrightpermission[1]{} 
\begin{document}

\title{Efficient Generation of Geographically Accurate Transit Maps}

\author{Hannah Bast}
\affiliation{%
  \institution{University of Freiburg}
  \city{Freiburg}
  \country{Germany}
}
\email{bast@cs.uni-freiburg.de}

\author{Patrick Brosi}
\affiliation{%
  \institution{University of Freiburg}
  \city{Freiburg}
  \country{Germany}
}
\email{brosi@cs.uni-freiburg.de}

\author{Sabine Storandt}
\affiliation{%
  \institution{JMU W\"urzburg}
  \city{W\"urzburg}
  \country{Germany}
}
\email{storandt@informatik.uni-wuerzburg.de}


\begin{abstract}
We present LOOM (Line-Ordering Optimized Maps), a fully automatic generator of geographically accurate transit maps.
The input to LOOM is data about the lines of a given transit network, namely for each line, the sequence of stations it serves and the geographical course the vehicles of this line take. We parse this data from GTFS, the prevailing standard for public transit data.
LOOM proceeds in three stages:
	(1) construct a so-called line graph, where edges correspond to segments of the network with the same set of lines following the same course;
	(2) construct an ILP that yields a line ordering for each edge which minimizes the total number of line crossings and line separations;
	(3) based on the line graph and the ILP solution, draw the map.
As a naive ILP formulation is too demanding, we derive a new custom-tailored formulation which requires significantly fewer constraints. Furthermore, we present engineering techniques which use structural properties of the line graph to further reduce the ILP size. For the subway network of New York, we can reduce the number of constraints from 229,000 in the naive ILP formulation to about 4,500 with our techniques, enabling solution times of less than a second. Since our maps respect the geography of the transit network, they can be used for tiles and overlays in typical map services. Previous research work either did not take the geographical course of the lines into account, or was concerned with schematic maps without optimizing line crossings or line separations.
\end{abstract}

\def\UrlFont{\normalsize}
\maketitle
\urlstyle{same}
\section{Introduction}\label{SEC:intro}

Cities with a public transit network usually have a map which illustrates the network and which is posted at all stations. Many map services also feature a transit layer where all lines and stations in an area are displayed.
Such a map should satisfy the following main criteria:\vspace{0pt}

\begin{enumerate}[parsep=0.5mm, wide, labelwidth=0mm, itemindent=2.3mm]
  \setlength\itemsep{0pt}
\item  It should depict the topology of the network: which transit lines are offered, which stations do they serve in which order, and which transfers are possible.
\item It should be neatly arranged and esthetically pleasing.
\item It should reflect the geographical course of the lines, at least to some extent.
\end{enumerate}

\smallskip\noindent
So far, such maps have been designed and drawn by hand.
Concerning (3), the designers usually take some liberty, either to make the map fit into a certain format or to simplify the layout, or both.
\begin{figure}[ht]
\centering
\vspace{1em}
\includegraphics[width=0.23\textwidth]{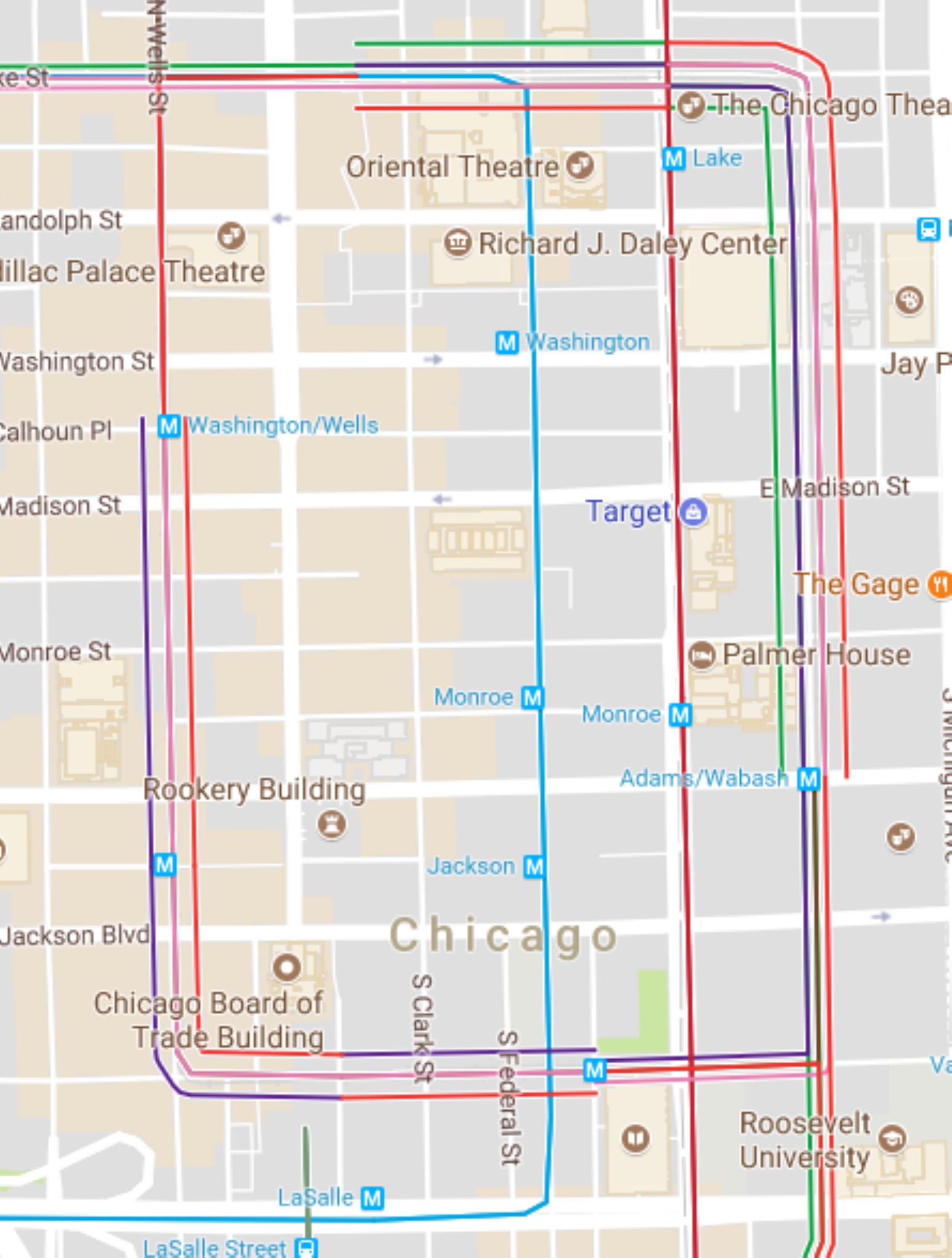}
\hfill
\includegraphics[width=0.23\textwidth]{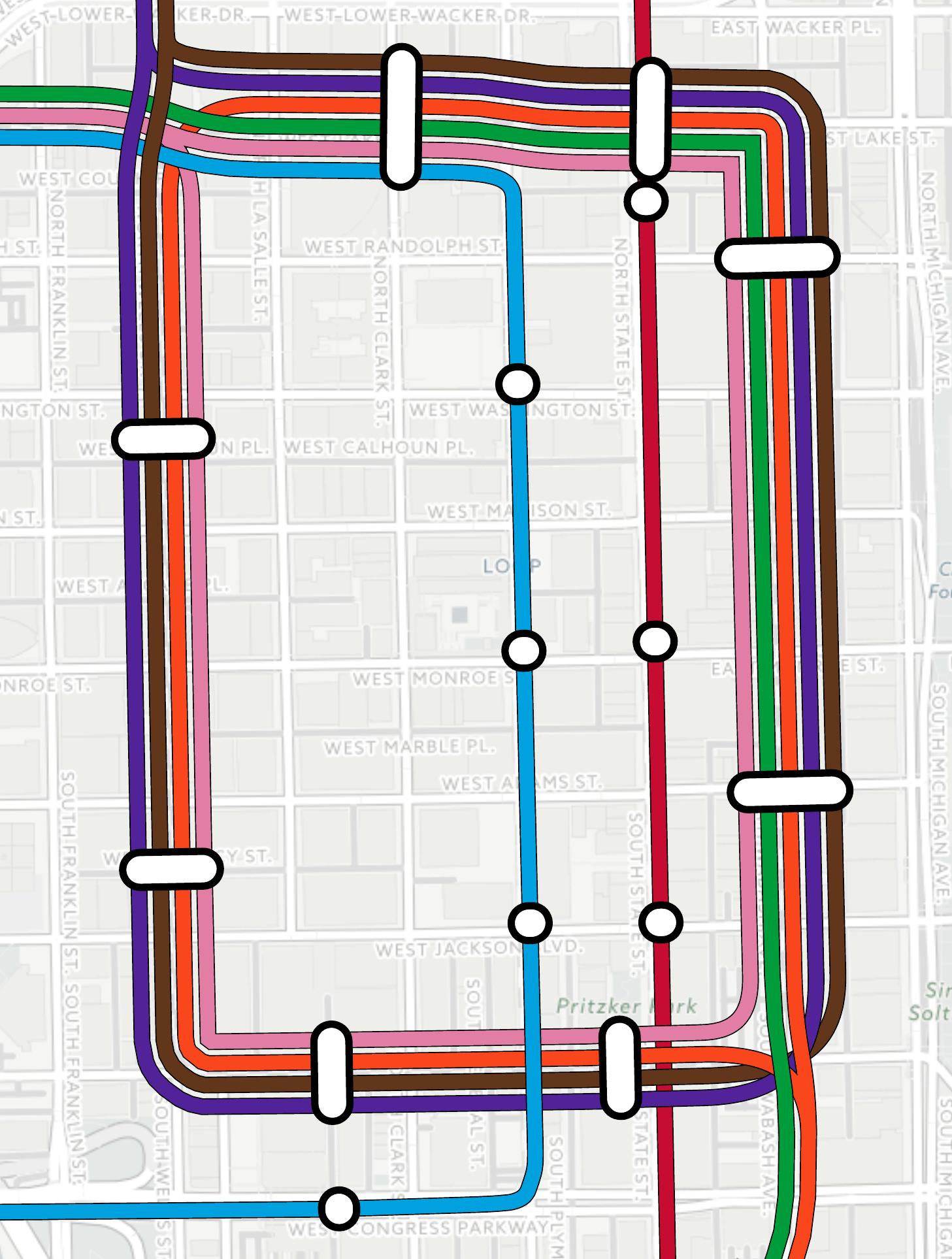}
\caption{Left: Google transit map cutout for Chicago. Right: LOOM map for the same area.}
\vspace{-1em}
\end{figure}

\begin{figure*}
  \centering
  \includestandalone[trim={19pt, 3pt, 0, 0}, clip, width=0.37\textwidth]{render_examples/vvs_transit_graph}
  \hspace{2cm}
  \includegraphics[trim={0cm 0 2.47cm 4.07cm},clip,width=0.37\textwidth]{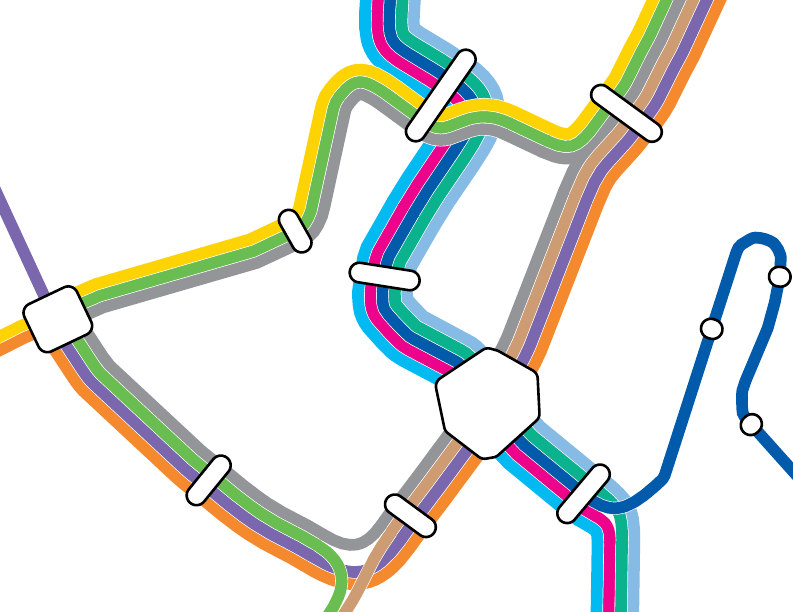}
  \caption{Left: Excerpt from a line graph which LOOM constructs for the 2015 light rail network of the city of Stuttgart from the given GTFS data. Each edge corresponds to a segment of the network where the same set of lines takes the same geographical course. Segment boundaries are often station nodes (large) but may also be intermediate nodes (small). The line ids for each segment are given in ascending order. LOOM's central optimization step computes a line ordering that determines how the lines are drawn in the map, and where line crossings and separations occur. Right: The corresponding excerpt from LOOM's transit map.}
  \label{FIG:transitgraphvvs}
\end{figure*}
The goal of this paper is to produce transit maps fully automatically, adhering to (3) rather strictly: within a given tolerance, the lines on the map should be drawn according to their geographical course. This rises several algorithmic challenges; in particular because the geographical course of some lines may overlap partially. These lines should then of course not be rendered on top of each other as this would obfuscate the visibility. Instead, they should be drawn next to each other. This requires to first identify overlapping parts and then to choose the line ordering in the rendered map. A bad ordering can lead to many unnecessary line crossings. Hence our goal is to find orderings that minimize these undesired crossings. As the number of possible orderings exceeds an octillion even for the transit network of medium sized cities, we need to develop efficient methods to find the best ordering in reasonable time.

\def\Hl{L}

\subsection{Overview and Definitions}\label{SEC:intro:definition}

LOOM proceeds in three stages, which we briefly describe in the following along with some notation and terminology that will be used throughout the paper.
Each stage is described in more detail in one of the following sections.

\smallskip\noindent
{\bf Input:}
The input to LOOM is a set ${\cal S}$ of stations and a set ${\cal L}$ of lines.
Each station has a geographical location.
Each line has a unique ID (in our examples: numbers), the sequence of stations it serves, and the geographical course between them.
This data is usually provided as part of a network's GTFS feed.


\smallskip\noindent
{\bf Line graph construction (Sect. \ref{SEC:graph}):}
In its first stage, LOOM computes a \emph{line graph}.
This is an undirected labeled graph $G = (V, E, \Hl)$, where $V \supseteq {\cal S}$ (each station is a node, but there may be additional nodes), $E$ is the set of edges, and each $e \in E$ is labeled with a subset $\Hl(e) \subseteq {\cal L}$ of the lines.
Intuitively, each edge corresponds to a segment of the network, where the same set of lines takes the same geographical course (within a certain tolerance), and there is a node wherever such a set of lines splits up in different directions.
Figure~\ref{FIG:transitgraphvvs} shows the line graph for an excerpt from the light rail network of Stuttgart.
We will see that the complexity of our algorithms in Sect.~\ref{SEC:ordering} depends on $M = \max_{e \in E} |\Hl(e)|$, the maximal number of lines per segment.

\smallskip\noindent
{\bf Line ordering optimization (Sect.~\ref{SEC:ordering}):}
In its second stage, LOOM computes an \emph{ordering} of $\Hl(e)$ for each $e \in E$.
This ordering determines where line crossings and separations occur, and is hence critical for the final map appearance.
Previous research referred to the problem of minimizing crossings as the metro-line crossing minimization problem (MLCM), see Sect~\ref{SEC:related}. We call a strongly related problem the metro-line node crossing minimization problem (MLNCM) with an optional line separation penalty (MLNCM-S) and formulate a concise Integer Linear Program (ILP) to solve it.

\smallskip\noindent
{\bf Rendering (Sect.~\ref{SEC:rendering}):}
In its third stage, LOOM draws the transit map based on the line graph from stage 1 and the ordering from stage 2.
Each station node $v$ is drawn as a polygon, where each side of the polygon corresponds to exactly one incident edge of $v$.
We call this side the \emph{node front} of that edge at that node.
The node front for an edge $e$ has $|\Hl(e)|$ so-called \emph{ports} (Fig.~\ref{FIG:crossings}).
Drawing the map then amounts to connecting the ports (according to the ordering computed in stage 2) and drawing the station polygons.
Figure \ref{FIG:transitgraphvvs}, right, shows a rendered transit map after layout optimization.

\subsection{Contributions}\label{SEC:intro:contrib}
\begin{itemize}[topsep=0pt, parsep=0.5mm,leftmargin=0mm,itemindent=4mm,itemsep=-1pt]
\renewcommand\labelitemi{$\bullet$}
\item We present a new automatic map generator, called LOOM (Line-Ordering Optimized Maps), for geographically accurate transit maps. The input is basic schedule data as provided in a GFTS feed.
 This is, as far as we know, the first research paper on this problem in its entirety. Previous research work considers only parts of this problem (oblivious either to the geographical course or to the order of the lines) and does not yield maps that can be used for tiles and overlays in typical map services.

\item We describe a line-sweeping approach to extract the line graph from a set of (partially overlapping) vehicle trips as they occur in real-world schedule data.
\item We phrase the crossing minimization problem in a novel way and provide an ILP formulation to solve it. Our new model resolves some issues with previous models, in particular, the restricted applicability of some algorithms to planar graphs, and the necessity of artificial grouping of crossings (which happens naturally with our approach).

\item As a naive ILP formulation turns out to lead to impractically many constraints, we derive an alternative formulation yielding significantly smaller ILPs in theory and practice.

\item We describe engineering techniques which allow to further simplify the line graph and hence lead to even smaller ILPs without compromising optimality of the final result.

\item We evaluate LOOM on the transit network of six cities around the world.
For each city, line graph construction, ILP solution and rendering together take less than 1 minute.

\item Our maps are publicly available online\footnote{\url{http://loom.informatik.uni-freiburg.de}}.
\end{itemize}

\subsection{Related Work}\label{SEC:related}

Previous research on the metro-line crossing minimization problem (MLCM), as briefly summarized in the following, typically comes without experimental evaluations and without the production of actual maps.
The problem of minimizing intra-edge crossings in transit maps was introduced in~\cite{ben06},
with the premises of not hiding crossings under station markers for esthetic reasons.
A polynomial time algorithm for the special case of optimizing the layout along a single edge was described.
The term MLCM was coined in~\cite{bek07}.
In that paper, optimal layouts for path and tree networks were investigated but arbitrary graphs were left as an open problem.
In~\cite{arg08, arg10, nol09}, several variants of MLCM were defined and efficient algorithms were presented for some of these variants, often with a restriction to planar graphs.
In \cite{asq08}, an ILP formulation for MLCM under the periphery condition (see Sect.~\ref{SEC:separation}) was introduced.
The resulting ILP was shown to have a size of $\mathcal{O}(|L|^2|E|)$ with $L$ being the set of lines and $E$ the set of edges in the derived graph.
In~\cite{fin13b}, it was observed that many (unavoidable) crossings scattered along a single edge are also not visually pleasing,
and hence crossings were grouped into so-called block crossings.
The problem of minimizing the number of block crossings was shown to be NP-hard on simple graphs just like the original MLCM problem \cite{fin13a}.
Our adapted MLNCM problem has the same complexity as MLCM and is hence also NP-hard.

Our line graph construction is related to edge bundling. The goal of edge bundling in general networks is to group edges in order to save ink when drawing the network. Usually, the embedding of the edges is not fixed a priori but can be chosen such that many bundles occur (possibly respecting side constraints as edges being short). For example, in \cite{hol09} a force-directed heuristic was described where edges attract other edges to form bundles automatically. Opposed to this, we are not allowed to embed edges arbitrarily as we want to maintain the geographical course of the vehicle trajectories. In \cite{pup11}, edge bundling in the context of metro line map layout was discussed, also considering orderings within the bundles to minimize crossings. But for their approach to work, the underlying graph has to fulfill a set of restrictive properties. For example, the so called \emph{path terminal property} demands that a node in the graph cannot be an endpoint of one line and an intermediate node of another line at the same time. But this structure regularly appears in real-world instances. For example, a local train might end at the main station of a town, while a long-distance train might have this station only as an intermediate stop. Also self-intersections are forbidden which excludes instances with cyclic subway lines. With these additional properties required in \cite{pup11} the problem becomes significantly easier but is no longer applicable to most real-world instances.

Another line of research focuses on drawing \emph{schematic} metro maps, for example, by restricting the representation of transit lines to octilinear polylines~\cite{hon06} or B\'ezier Curves~\cite{fin12}. See also~\cite{nol14} for a recent survey on automated metro map layout methods. These approaches strongly abstract from the geographical course of the lines (and often also from station positions), and the minimization of line crossings or separations is not part of the problem.
In particular, the resulting maps cannot be used for tiles or overlays in typical map services.

There is also some applied work on transit maps, but without publications of the details.
One approach that seems to use a model similar to ours was described by Anton Dubreau in a blog post~\cite{dub16} although without a detailed discussion of their method. As far as we are aware there are no papers on MLCM concerned with real public transit data.

\section{Line Graph Construction}\label{SEC:graph}

This section describes stage 1 of LOOM: given line data, construct the line graph.
We assume that the data is given in the GTFS format \cite{gtfs}.
In GTFS, each trip (that is, a concrete tour of a vehicle of a line) is given explicitly and the graph induced by this data has many overlapping edges that may (partially) share the same path.

Let $e_1, e_2$ be two edges in $G$ with their geometrical paths $\tau_{e_1}$ and $\tau_{e_2}$. We define a parametrization $\tau(t): [0,1] \mapsto \mathbb{R}^2$ which maps the progress $t$ to a point $p \in \mathbb{R}^2$ on $\tau$. To decide whether a segment of $\tau_{e_1}$ is similar to a segment of $\tau_{e_2}$, we use a simple approximation. For some distance threshold $\hat{d}$, we say $\left(\left(t_1, t_2\right), \left({t'}_1, {t'}_2\right)\right)$ is a shared segment of $e_1$ and $e_2$ if
\begin{equation}
	\forall u \in [t_1, {t'}_1]: \exists u' \in [t_2, {t'}_2]:\left\|\tau_{e_1}\left(u\right) - \tau_{e_2}\left(u'\right)\right\| \leq \hat{d}.
\end{equation}

We transform $G$ into a line graph $G'$ by repeatedly combining shared segments between two edges $e_1 = \{u_1, v_1\}$ and $e_2 = \{u_2, v_2\}$ into a single new edge $e_{12}$ until no more shared segments can be found. The new path $\tau_{e_{12}}$ is averaged from the shared segments on $\tau_{e_1}$ and $\tau_{e_2}$. Two new non-station nodes $u'$ and $v'$ are introduced and split $e_1$ and $e_2$ such that $e_1 = \{u_1, u'\}$, $e_2 = \{u_2, u'\}$, $e'_1 = \{v', v_1\}$, $e'_2 = \{v', v_2\}$ and $e_{12} = \{u', v'\}$. Note that the new non-station nodes $v'$ and $u'$ will always have a degree of 3.

To find the shared segments between $\tau$ and $\tau'$, we sweep over $\tau$ in $n$ steps of some $\Delta t$, measuring the distance $d$ between $\tau(i\cdot\Delta t)$ and $\tau'$ at each $i < n$ along the way. If $d \leq \hat{d}$, we start a new segment. If $d > \hat{d}$ and a segment is open, we close it. The algorithm can be made more robust against outliers by allowing $d$ to exceed $\hat{d}$ for a number of $k$ steps. It can be sped up by indexing every linear segment of every path in a geometric index (for example, an R-Tree).

\section{Line Ordering Optimization}\label{SEC:ordering}

This section describes stage 2 of LOOM, namely how to solve MLNCM: given a line graph, compute an ordering of the lines for each edge such that the total number of  crossings in the final map is minimized. Contrary to the classic MLCM problem, which imposes a right and left ordering on each $L(e)$ and allows crossings to occur anywhere on $e$, MLNCM only imposes a single ordering on each edge and restricts crossing events to nodes. This will proof advantageous during rendering, see Sect.~\ref{SEC:rendering}. As the set of stations $\cal S$ is only a subset of $V$ in our model (Sect.~\ref{SEC:intro:definition}), we can still avoid line crossings in them.

For each edge $e$, there are $|\Hl(e)|!$ many orderings, therefore the total number of combinations for the whole graph is immense.
We formulate an ILP to find an optimal solution.
We first define a baseline ILP which explicitly considers line crossings and has $\mathcal{O}(|E|M^{2})$ variables and $\mathcal{O}(|E|M^{6})$ constraints.
We then define an improved ILP with only $\mathcal{O}(|E|M^2)$ constraints and which also considers line separations (MLNCM-S).

\def\Hsum{\sum\nolimits}
\subsection{Baseline ILP}\label{SEC:baseline}
For every edge $e \in E$, we define $\left|L\left(e\right)\right|^{2}$ decision variables $x_{elp} \in \{0,1\}$ where $e$ indicates the edge, $l \in L(e)$ indicates the line, and $p=1, ..., \left|L\left(e\right)\right|$ indicates the position of the line in the edge. We want to enforce $x_{elp}=1$ when line $l$ is assigned to position $p$, and $0$ otherwise. This can be realized with the following constraints:
\begin{equation}
\forall l \in L(e): \Hsum_{p=1}^{L(e)} x_{elp}=1.  \label{EQ:p_constr1_bl}
\end{equation}
To ensure that exactly one line is assigned to each position, we need the following additional constraints:
\begin{equation}
	\forall p \in  \left\{1,...,\left|L\left(e\right)\right|\right\} : \Hsum_{l \in L(e)} x_{elp} = 1. \label{EQ:p_constr2_bl}
\end{equation}
\begin{figure}
  \centering
	$\vcenter{\hbox{\includegraphics[width=0.47\textwidth]{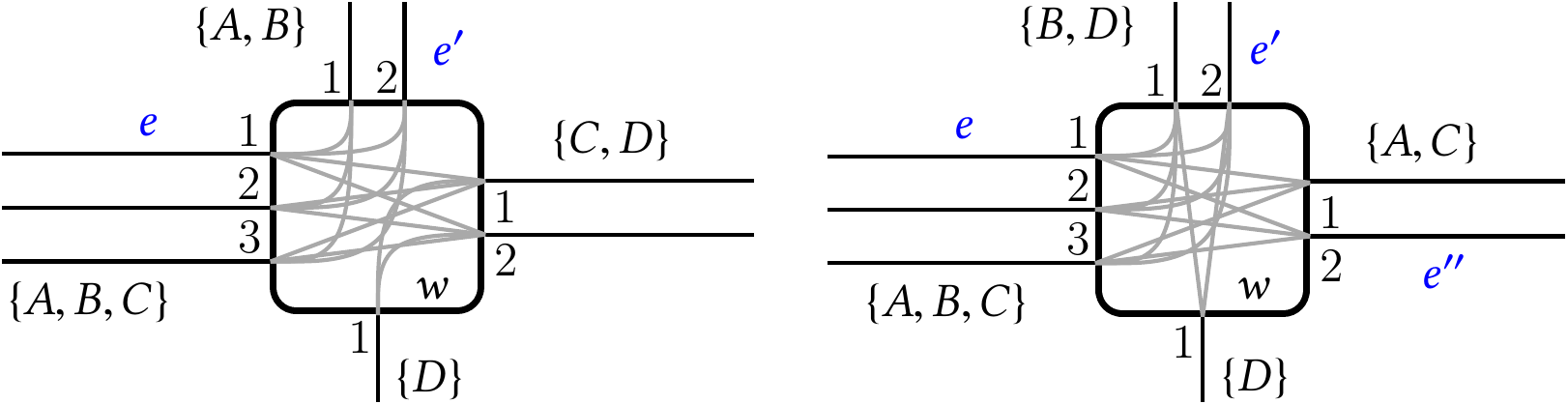}}}$
	\caption{Example instances. Both station polygons have 4 \emph{node fronts}, each corresponding to an incident edge. Each node front has exactly one port ($1, 2, ...$) for each line traversing through its edge.}
	\label{FIG:crossings}
\end{figure}Let $A, B$ be two lines belonging to an edge $e = \{v,w\}$ and both extend over $w$.
We distinguish two cases: either $A$ and $B$ continue along the same adjacent edge $e'$ (Fig.~\ref{FIG:crossings}, left), or they continue along different edges $e'$ and $e''$ (Fig.~\ref{FIG:crossings}, right).

In the first case, $A$ and $B$ induce a crossing if the position of $A$ is smaller than the position of $B$ in $L(e)$, so $p_e(A) < p_e(B)$, but vice versa in $L(e')$. We introduce the decision variable $x_{ee'AB} \in \{0,1\}$, which should be $1$ in case a crossing is induced and $0$ otherwise. To enforce this, we create one constraint per possible crossing. For example, a crossing would occur if we have $p_e(A)=1$ and $p_e(B)=2$ as well as $p_{e'}(A)=2$ and $p_{e'}(B) =1$. We encode this as follows:
\begin{align}
	x_{eA1} + x_{eB2} + x_{e'A2} + x_{e'B1} - x_{ee'AB} \leq 3  \label{EQ:crossdec_constr_bl}.
\end{align}
In case the crossing occurs, the first four variables are all set to 1. Hence their sum is 4 and the only way to fulfill the $\leq 3$ constraint is to set $x_{ee'AB}$ to $1$. In the example given in Fig.~\ref{FIG:crossings}, six such constraints are necessary to account for all possible crossings of the lines $A$ and $B$ at node $w$.
The objective function of the ILP then minimizes the sum over all variables $x_{ee'AB}$.

In the second case, the actual positions of $A$ and $B$ in $e'$ and $e''$ do not matter, but just the order of $e'$ and $e''$. We introduce a split crossing decision variable $x_{ee'e''AB} \in \{0,1\}$ and constraints of the form $x_{eAi} + x_{eBj} - x_{ee'e''AB} \leq 1$ for all orders of $A$ and $B$ at $e$ with $i < j$ as in that case a crossing would occur.
We add $x_{ee'e''AB}$ to the objective function.

For mapping lines to positions at each edge we need at most $|E|M^{2}$ variables and $2|E|M$ constraints. To minimize crossings, we have to consider at most $M^{2}$ pairs of lines per edge, and introduce a decision variable for each such pair. That makes at most $|E| M^{2}$ additional variables, which all appear in the objective function. Most constraints are introduced when two lines continue over a node in the same direction. In that case, we create no more than $\binom{M}{2}^{2} < M^{4}$ constraints per line pair per edge, so at most $|E| M^{6}$ in total.
In summary, we have $\mathcal{O}(|E|M^{2})$ variables and $\mathcal{O}(|E|M^{6})$ constraints.

\subsection{Improved ILP Formulation}\label{SEC:improved}
The $\mathcal{O}(|E|M^{2})$ variables in the baseline ILP seem to be reasonable, as indeed $\Omega(|E|M^{2})$ crossings could occur. But the $\mathcal{O}(|E|M^{6})$ constraints are due to enumerating all possible position inversions explicitly. If we could check the statement \emph{position of A on $e$ is smaller than the position of B} efficiently, the number of constraints could be reduced. To have such an oracle, we first modify the line-position assignment constraints.

Instead of a decision variable encoding the exact position of a line, we now use $x_{el\leq p} \in \{0,1\}$ which is $1$ if the position of $l$ in $e$ is $\leq p$ and $0$ otherwise. To enforce a unique position, we use the constraints:
\begin{align}
	\forall l \in L(e)~ \forall p \in \{1, ..., \left|L\left(e\right)\right|-1\}: \quad x_{el\leq p} \leq x_{el\leq p+1}. \label{EQ:up_constr}
\end{align}
This ensures that the sequence can only switch from $0$ to $1$, exactly once. To make sure that at some point a $1$ appears and that each position is occupied by exactly one line, we additionally introduce the following constraints:
\begin{align}
	\forall p \in \{1, ..., \left|L\left(e\right)\right|\}: \quad  \Hsum_{l \in L(e)} x_{el\leq p} = p. \label{EQ:01_constr}
\end{align}
So for exactly one line $l$, $x_{el\leq 1} = 1$, for exactly two lines $l'$ and $l''$, $x_{el'\leq 2} = x_{el''\leq 2} = 1$ (where for one $l \in \{l',l''\}$, $x_{el\leq 1} =1$) and so on.

We reconsider the example in Fig.~\ref{FIG:crossings}, left. Before, we enumerated all possible positions which induce a crossing for $A, B$ at the transition from $e$ to $e'$. But it would be sufficient to have variables which tell us whether the position of $A$ is smaller than the position of $B$ in $e$, and the same for $e'$, and then compare those variables. For a line pair $(A,B)$ on edge $e$ we call the respective variables $x_{eB<A}, x_{eA<B} \in \{0, 1\}$. 
To get the desired value assignments, we add the following constraints:
\begin{gather}
	\Hsum_{p=1}^{\left|L\left(e\right)\right|} x_{eA\leq p} - \Hsum_{p} x_{eB\leq p} + x_{eB<A} M \geq 0 \\
	x_{eB<A} + x_{eA<B}=1.
\end{gather}
The equality constraints make sure that not both $x_{eA<B}$ and $x_{eB<A}$ can be $1$. If the position of $A$ is smaller than the position of $B$, then more of the variables corresponding to $A$ are $1$, and hence the sum for $A$ is higher. So if we subtract the sum for $B$ from the sum for $A$ and the result is $\geq 0$, we know the position of $A$ is smaller and $x_{eB<A}$ can be $0$. Otherwise, the difference is negative, and we need to set $x_{eB<A}$ to $1$ to fulfill the inequality. It is then indeed fulfilled for sure as the position gap can never exceed the number of lines per edge.

To decide if there is a crossing, we would again like to have a decision variable $x_{ee'AB} \in \{0,1\}$ which is $1$ in case of a crossing and $0$ otherwise. The constraint
\begin{gather}
	\left|x_{eA<B}-x_{e'A<B}\right| - x_{ee'AB} \leq 0
\end{gather}
realizes this, as either $x_{eA<B} = x_{e'A<B}$ (both $0$ or both $1$) and then $x_{ee'AB}$ can be $0$, or they are not equal and hence the absolute value of their difference is $1$, enforcing $x_{e'AB}=1$. As absolute value computation cannot be part of an ILP we use the following equivalent standard replacement:
\begin{align}
	x_{eA<B} - x_{e'A<B} - x_{ee'AB} &\leq 0 \label{EQ:abs_upper}\\
	-x_{eA<B} + x_{e'A<B} - x_{ee'AB} &\leq 0. \label{EQ:abs_lower}
\end{align}

For the line-position assignment, we need at most $|E|M^{2}$ variables and constraints just like before.
For counting the  crossings, we need a constant number of new variables and constraints per pair of lines per edge.
Hence the total number of variables and constraints in the improved ILP is $\mathcal{O}(|E|M^{2})$.

\begin{figure*}[t]
\vspace{-.5em}
\centering
\begin{minipage}{.472\textwidth}
  \centering
	\includegraphics[trim={2.25 2.2 2.25 2.2},clip,width=.32\textwidth]{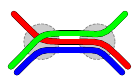}\hspace{28pt}
	\includegraphics[trim={2.25 2.2 2.25 2.2},clip,width=.32\textwidth]{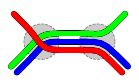}
	\caption{Minimized crossings in the left example, but the right example better indicates line pairings.}
	\label{FIG:linesplitting}
\end{minipage}%
\hfill
\begin{minipage}{.472\textwidth}
  \centering
	\includegraphics[trim={2.25 2.2 2.25 2.2},clip,width=.32\textwidth]{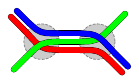}\hspace{28pt}
	\includegraphics[trim={2.25 2.2 2.25 2.2},clip,width=.32\textwidth]{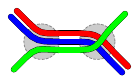}
	\caption{Both orderings have 2 crossings, but in the right example they are done in one pass.}
	\label{FIG:linesplitting2}
\end{minipage}
\vspace{-1em}
\end{figure*}

\subsection{Preventing Line Partner Separation}\label{SEC:separation}

So far, we have only considered the number of crossings.
Another relevant criterion for esthetic appeal is that ``partnering'' lines are drawn side by side.
Fig.~\ref{FIG:linesplitting} and Fig.~\ref{FIG:linesplitting2} provide two examples. We address this by punishing line separations and call this extension to our original MLNCM problem \mbox{MLNCM-S}. For two adjacent edges $e$ and $e'$ and a line pair $(A, B)$ that continues from $e$ to $e'$, if $A$ and $B$ are placed alongside in $e$ but not in $e'$, we want to add a penalty to the objective function. For this, we add a variable $x_{eA\|B} \in \{0, 1\}$ which should be $0$ if $\left|p_{e}(A) - p_{e}(B)\right| = 1$ (if they are partners in $e$) and $1$ otherwise. As $x_{eA\|B} = x_{eB\|A}$, we define a set $U(e)$ of unique line pairs such that $(l, l') \in U(e) \Rightarrow (l', l) \not\in U(e)$. We add the following constraints per line pair $(A, B)$ in $U(e)$:
\begin{align}
	\Hsum_{p=1}^{\left|L\left(e\right)\right|} x_{eA\leq p} - \Hsum_{p} x_{eB\leq p} - x_{eA\|B} M &\leq 1 \label{EQ:sep_cstr_1} \\
	\Hsum_{p=1}^{\left|L\left(e\right)\right|} x_{eB\leq p} - \Hsum_{p} x_{eA\leq p} - x_{eA\|B} M &\leq 1. \label{EQ:sep_cstr_2}
\end{align}
If $|p_{e}(A) - p_{e}(B)| = 1$, then the sum difference is $\leq 1$ and $x_{eA\|B}$ can be 0. If $|p_{e}(A) - p_{e}(B)| > 1$, then either (\ref{EQ:sep_cstr_1}) or (\ref{EQ:sep_cstr_2}) enforce $x_{eA\|B} = 1$. To prevent the trivial solution where $x_{eA\|B}$ is always 1, we add the following constraint per edge $e$:
\begin{align}
	\Hsum_{(l, l') \in U(e)} x_{el\|l'} \leq \binom{\left|L\left(e\right)\right|}{2} - \left|L\left(e\right)\right| - 1,  \label{EQ:sep_cstr_3}
\end{align}
as there are $\binom{\left|L\left(e\right)\right|}{2}$ line pairs $(l, l') \in U(e)$ of which $\left|L\left(e\right)\right| - 1$ are next to each other.

Like in Sect.~\ref{SEC:improved}, we add a decision variable $x_{ee'A\|B}$ to the objective function that should be $1$ if $A$ and $B$ are separated between $e$ and $e'$ and $0$ otherwise:
\begin{align}
	x_{eA\|B} - x_{e'A\|B} - x_{ee'A\|B} &\leq 0 \\
	-x_{eA\|B} + x_{e'A\|B} - x_{ee'A\|B} &\leq 0.
\end{align}

As we only add 1 constraint per edge and a constant number of constraints and variables per line pair in each edge, the total number of variables and constraints remains $\mathcal{O}(|E|M^2)$.

Interestingly, punishing line separations also addresses a special case of the periphery condition introduced in \cite{asq08}. In general, this condition holds if lines ending in a station are always drawn at the left- or rightmost position in each incident edge. For nodes with degree $\leq 2$, the periphery condition is ensured in \mbox{MLNCM-S} (Fig.~\ref{FIG:periphery}, left). For other nodes, however, it is not guaranteed (Fig.~\ref{FIG:periphery}, right).

\begin{figure}[b]
\centering
\begin{minipage}{.48\textwidth}
  \centering
	\includegraphics[trim={0 1.22 0 1.8},clip,width=.49\textwidth]{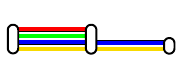}
	\hfill
	\includegraphics[trim={0 1.75 0 1.8},clip,width=.49\textwidth]{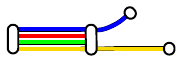}
	\caption{Left: Periphery condition guaranteed by separation penalty. Right: Periphery condition not guaranteed by separation penalty.}
	\label{FIG:periphery}
\end{minipage}
\vspace{-.4em}
\end{figure}

\subsection{Placement of Crossings or Separations}\label{SEC:crossing_placement}
The placement of crossings or separations may be fine-tuned by adding node-based weighting factors $w_\times(v)$ (for crossings) and $w_\|(v)$ (for separations) to the objective function to prefer nodes or to break ties. For example, $w_\times(v)$ may depend on the node degree.

As described above, we especially want to prevent crossings or separations in station nodes. This can be achieved by adding constant global weighting factors $w_{\cal S\times}$ and $w_{\cal S\|}$ to each $x_{ee'll'}$ and $x_{ee'l\|l'}$ in the objective function if $l$ and $l'$ continue over a node $v_s \in \cal S$. These factors have to be chosen high enough so that a crossing or separation in any other node $v \not\in \cal S$ is never more expensive than in $v_s$. As all $w_\times(v)$ and $w_\|(v)$ appear as coefficients in the objective function, they have to be invariant to the actual line orderings. We can thus determine the maximum possible costs $\hat w_\times$ and $\hat w_{\|}$ prior to optimization and choose $w_{\cal S\times} = \hat w_\times$ and $w_{\cal S\|} = \hat w_{\|}$.

\section{Core Graph Reduction}\label{SEC:coreprobgraph}
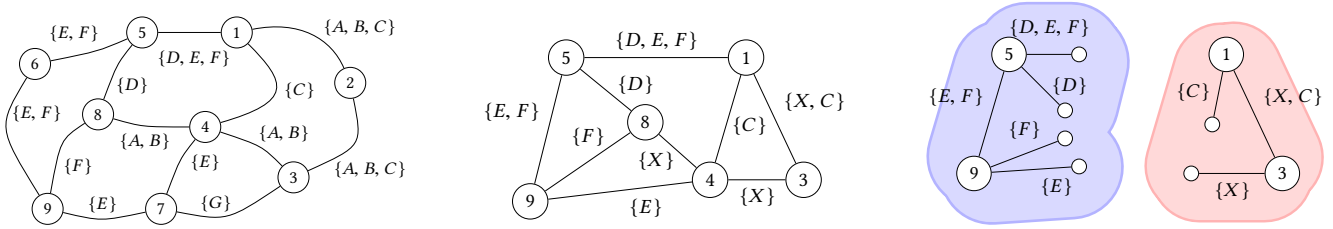
\begin{figure*}
\centering
		\tikzstyle{graphnode}=[minimum size=14pt, inner sep=0, circle, draw=black, solid, fill=white]
	\begin{tikzpicture}[font=\small]
		\foreach \place/\x in {{(9.5, 3.8)/1}, {(11.3,3)/2},{(10.4,1.5)/3},{(9,2.3)/4}, {(8, 3.8)/5}, {(6.3, 3.3)/6}, {(8.3, 1)/7}, {(7.3, 2.5)/8}, {(6.5, 1)/9}}
		\node[graphnode] (\x) at \place {$\x$};

		\draw (1) .. controls (10.5,4) and (11, 3.7) .. (2) node [xshift=2,yshift=10, midway, below right] {$\{A,B,C\}$};
		\draw (2) .. controls (11.5,2) and (11.5, 2).. (3) node [midway, below right, xshift=-13, yshift=-5] {$\{A,B,C\}$};
		\draw (3) .. controls (9.9, 2) .. (4) node [midway, above right, yshift=-1, xshift=-4] {$\{A, B\}$};
		\draw (1) .. controls (10.3, 3) and (10.4, 2.7) .. (4) node [midway, right] {$\{C\}$};
		\draw (5) .. controls (9, 3.8) .. (1) node [xshift=-3.5, yshift=-5, midway, below] {$\{D, E, F\}$};
		\draw (4) .. controls (8, 2.3) .. (8) node [midway, below] {$\{A, B\}$};
		\draw (4) .. controls (8.6, 1.8) .. (7) node [yshift=-1, xshift=1,midway, right] {$\{E\}$};
		\draw (7) .. controls (7.35, .8) .. (9) node [midway, above] {$\{E\}$};
		\draw (9) .. controls (6.7, 2.1) .. (8) node [yshift=-3, xshift=-3, midway, below right] {$\{F\}$};
		\draw (6) .. controls (5.7, 2.3) .. (9) node [yshift=7, xshift=-1, midway, right] {$\{E,F\}$};
		\draw (6) .. controls (7.3, 3.5) .. (5) node [yshift=1, xshift=-10, midway, above] {$\{E,F\}$};
		\draw (8) .. controls (7.6, 3.3) .. (5) node [yshift=-2, xshift=-4, midway, below right] {$\{D\}$};
		\draw (7) .. controls (9.35, .8) .. (3) node [xshift=-5.5, yshift=-1, midway, above] {$\{G\}$};
	\end{tikzpicture}
	\hfill
		\tikzstyle{graphnode}=[minimum size=14pt, inner sep=0, circle, draw=black, solid, fill=white]
	\begin{tikzpicture}[font=\small]
		\newcommand{\highlight}[3]{
			\path [line cap = round, line join = round, line width = #1, draw = #2!30] #3;
			\pgfmathsetmacro{\innerlinewidth}{0.95 * #1}
			\path [line cap = round, line join = round, line width = \innerlinewidth, draw = #2!15] #3;
		}

		\foreach \place/\x in {{(9.5,4)/1}, {(10.3,2.3)/3},{(9,2.3)/4}, {(7, 4)/5}, {(8.1, 3.1)/8}, {(6.5, 2)/9}}
		\node[graphnode] (\x) at \place {$\x$};

		\draw (1) -- (3) node [midway, above right] {$\{X,C\}$};
		\draw (3) -- (4) node [midway, below] {$\{X\}$};
		\draw (1) -- (4) node [midway, below right, xshift=-2pt, yshift=4pt] {$\{C\}$};
		\draw (5) -- (1) node [midway, above] {$\{D,E,F\}$};
		\draw (4) -- (8) node [midway, below left, yshift=3pt, xshift=2pt] {$\{X\}$};
		\draw (4) -- (9) node [midway, below right] {$\{E\}$};
		\draw (9) -- (8) node [midway, above, yshift=3pt] {$\{F\}$};
		\draw (5) -- (9) node [midway, above left] {$\{E,F\}$};
		\draw (8) -- (5) node [midway, right, xshift=1pt, yshift=2pt] {$\{D\}$};
	\end{tikzpicture}
	\hfill
		\tikzstyle{graphnode}=[minimum size=14pt, inner sep=0, circle, draw=black, solid, fill=white]
	\tikzstyle{connnode}=[minimum size=6pt, inner sep=0, circle, draw=black, solid, fill=white]
	\begin{tikzpicture}[font=\small]
		\newcommand{\highlight}[3]{
			\path [line cap = round, line join = round, line width = #1, draw = #2!30] #3;
			\pgfmathsetmacro{\innerlinewidth}{0.95 * #1}
			\path [line cap = round, line join = round, line width = \innerlinewidth, draw = #2!15] #3;
		}

		\foreach \place/\x in {{(10.1,3.7)/1}, {(10.9,2)/3}, {(7, 3.7)/5}, {(6.5, 2)/9}}
		\node[graphnode] (\x) at \place {$\x$};

		\foreach \place/\x in {{(8, 3.7)/11}, {(9.6,2)/4}, {(9.9,2.7)/41}, {(7.8, 2.9)/13}, {(7.8, 2.5)/14}, {(8,2.1)/15}}
		\node[connnode] (\x) at \place {};

		\draw (1) -- (3) node [midway, above right] {$\{X,C\}$};
		\draw (3) -- (4) node [midway, below] {$\{X\}$};
		\draw (1) -- (41) node [midway, left, xshift=0pt, yshift=1pt] {$\{C\}$};

		\draw (5) -- (9) node [midway, above left] {$\{E,F\}$};

		\draw (13) -- (5) node [midway, right, xshift=1pt, yshift=0pt] {$\{D\}$};
		\draw (9) -- (14) node [midway, above, yshift=3pt] {$\{F\}$};

		\draw (15) -- (9) node [midway, below right] {$\{E\}$};

		\draw (5) -- (11) node [midway, above, yshift=4pt] {$\{D,E,F\}$};

		\begin{pgfonlayer}{background}
			\highlight{15mm}{blue}{(15) -- (9) -- (5) -- (11) -- (13);}
			\highlight{14mm}{red}{(1) -- (3) -- (4) -- (1);}
		\end{pgfonlayer}
	\end{tikzpicture}
	\caption{Left: line graph $G$ with 7 lines. Middle: core graph of $G$ after applying pruning rules, $\{A, B\}$ was collapsed into $\{X\}$. Right: ordering-relevant connected components of $G$ after applying splitting rules.}
	\label{FIG:coreoptimgraph}
	\vspace{-.4em}
\end{figure*}
It is possible to further simplify the optimization problem. We make the following observations:
\begin{lemma}\label{LEM:linepairing}
If for some set $\mathcal{B} = \{A, B, C, ...\} \subseteq \mathcal{L}$ it holds for all $l \in \mathcal{B}, e \in E: l \in L(e) \Rightarrow \mathcal{B} \subseteq L(e)$, then the optimal ordering is to always bundle $A, B, C, ...$ next to each other with a fixed, global ordering.
\end{lemma}
\begin{proof}
Let $L \in \mathcal{B}$ be the line that induces the minimal number of crossings and separations for some solution $\sigma$. Since all $l \in \mathcal{B}$ take the exact same path through the network, a solution can only be better than or equal to $\sigma$ if it bundles all $l \neq L$ alongside $L$.
\end{proof}
\begin{lemma}\label{LEM:crossingmoveing}
	Given an optimal ordering for each $L(e)$. We say a node $v$ belongs to $W$ if $\text{deg}(v) = 2$ and for its adjacent edges $e$ and $e'$ the set of lines $L(e)$ is equal to $L(e')$. A crossing or a separation in some $v \in W$ can always be moved from $v$ to a node $v' \not\in W$ without negatively affecting optimality.
\end{lemma}
\begin{proof}
We set $L^* = L(e) = L(e')$ and first consider crossings. There are two possible cases:
\begin{enumerate*}
	\item all $l \in L^*$ always occur together in each edge. Then Lemma~\ref{LEM:linepairing} holds, and the ordering of $L(e)$ is the same as of $L(e')$, inhibiting any crossings in $v$. We can thus ignore this case.
	\item Lemma~\ref{LEM:linepairing} does not hold and the lines in $L^*$ separate in some node $v' \neq v$. Then they either diverge into separate edges at $v'$, or a subset of them ends in $v'$. If they diverge, the degree of $v'$ has to be at least 3, implicating $v' \not\in W$. If some (or one) of them end in $v'$, then $v'$ has to be adjacent to at least 2 edges $e, e'$ with $L(e) \neq L(e')$, again implicating $v' \not\in W$. Such a $v'$ will thus indeed always exist. Under a uniform crossing penalty, we can trivially move the crossing from $v$ to $v'$ without affecting optimality. Under the penalty described in Sect.~\ref{SEC:crossing_placement}, optimality will also not be affected negatively, because $\text{deg}(v)$ is always 2, implying that $v$ is a station (Sect.~\ref{SEC:graph}).
\end{enumerate*} The same argument holds for line separations.
\end{proof}
\begin{lemma}\label{LEM:termini}
If for some edge $e$ all $l \in L(e)$ end in a node $v$ or $|L(e)| = 1$, the ordering of $L(e)$ will not affect the number of orderings or separations in $v$.
\end{lemma}
\begin{proof}
In the first case, no $l \in L(e)$ extends over $v$, so they cannot introduce any crossing or separation. In the second case, all orderings of $L(e)$ are equivalent (there is only one).
\end{proof}

Using the above, we can simplify the input graph with the following pruning rules (Fig.~\ref{FIG:coreoptimgraph},~middle):
\begin{enumerate}[parsep=0.5mm, wide, labelwidth=0mm, itemindent=2.3mm]
	\setlength\itemsep{1pt}
	\item delete each node $v$ with degree 2 and $L(e) = L(e')$, and combine the adjacent edges $e = \{u, v\}$, $e' = \{v, w\}$ into a single new edge $ee' = \{u, w\}$ with $L(ee') = L(e) = L(e')$ (Lemma \ref{LEM:crossingmoveing}).
	\item collapse lines that always occur together into a single new line $k$ (Lemma~\ref{LEM:linepairing}). Weight crossings with $k$ by the number of lines it combines to avoid distorting penalties.
	\item remove each edge $e = \{u, v\}$ where $u$ and $v$ are termini for all $l \in L(e)$ (Lemma~\ref{LEM:termini}).
\end{enumerate}
\noindent
This core graph of $G$ may be further broken down into ordering-relevant connected components using the rules below (Fig.~\ref{FIG:coreoptimgraph},~right). The components can then be optimized separately.
\begin{enumerate}[parsep=0.5mm, wide, labelwidth=0mm, itemindent=2.3mm]
	\setlength\itemsep{1pt}
    \item cut each edge $e = \{u, v\}$ with $\left|L\left(e\right)\right| = 1$ into two edges $e' = \{u, v'\}$ and $e'' = \{v'', v\}$ (Lemma~\ref{LEM:termini}).
    \item replace each edge $e = \{u, v\}$ where $v$ has a degree $>1$ and is a terminus node for each $l \in L(e)$ with an edge $e' = \{u, v'\}$ where $v'$ is only connected to $e'$ (Lemma~\ref{LEM:termini}).
\end{enumerate}
\section{Rendering}\label{SEC:rendering}
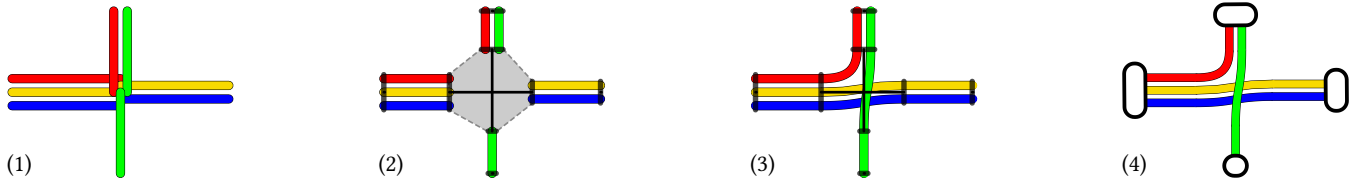
\begin{figure*}
  \centering
  	\scalebox{1} {
	  	\begin{picture}(80, 56)
			\put(0,0){\includegraphics[width=0.17\textwidth]{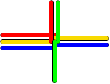}}
			\put(0,3){(1)}
		\end{picture}
	}
	\hfill
	\scalebox{1} {
	  	\begin{picture}(80, 56)
			\put(0,0){\includegraphics[width=0.17\textwidth]{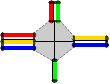}}
			\put(0,3){(2)}
		\end{picture}
	}
    \hfill
    \scalebox{1} {
	  	\begin{picture}(80, 56)
			\put(0,0){\includegraphics[width=0.17\textwidth]{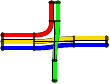}}
			\put(0,3){(3)}
		\end{picture}
    }
    \hfill
    \scalebox{1} {
	    \begin{picture}(80, 56)
			\put(0,0){\includegraphics[width=0.17\textwidth]{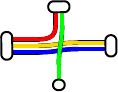}}
			\put(0,3){(4)}
		\end{picture}
	}
	\caption{The four steps of rendering a given line graph: (1) render ordered lines as edges, (2) free node area, (3) render inner connections, (4) render station overlays.}
	\label{FIG:renderingproc}
\end{figure*}
This section describes stage 3 of LOOM: given the line graph as computed in stage 1, and a line ordering for each edge as computed in stage 2, render the actual map.
We split this into four basic steps, as illustrated in Fig.~\ref{FIG:renderingproc}.

In the first step (1), we make use of the fact that only a single ordering is imposed on each $L(e)$ and render each $l \in L(e)$ by perpendicular offsetting the edge's geometry $\tau_{e}$ by $-w\left|L(e)\right|/2 + w\left(p_{e}(l)-1\right)$, where $w$ is the desired line width. In the next step (2), we make room for the line connections between edges by expanding the node fronts (and thus the node polygon). As a stopping criteria for this expansion, we simply use a maximum distance from the node front to its original position. The line connections in the node are then rendered by connecting all port pairs (3). In our experiments, we used cubic B\'ezier curves, but for schematic maps a circular arc or even a straight line might be preferable.

For the station rendering (4), we found that the buffered node polygon already yields reasonable results, although with much potential for improvement. We also experimented with rotating rectangles until the total sum of the deviations between each node front orientation and the orientation of the rectangle was minimized. Both approaches can be seen in Fig.~\ref{FIG:transitgraphvvs}.

\section{Evaluation}
We tested LOOM on the public transit schedules of six cities: Freiburg, Dallas, Chicago, Stuttgart, Turin and New York\footnote{with uncollapsed express lines}.
Table~\ref{TBL:datasets} provides the basic dimensions of each dataset and the time needed to extract the line graph.

\def\degv{\text{deg}(v)}
\def\Hs{\makebox[1.6mm][l]{\hspace{0.2mm}\footnotesize s}}
\def\Hm{\makebox[1.6mm][l]{\hspace{0.2mm}\footnotesize m}}
\def\Hh{\makebox[1.6mm][l]{\hspace{0.2mm}\footnotesize h}}
\def\Hhline{\\[.7mm]\hline}
\begin{table}
  \caption[]{Graph dimensions for our datasets Freiburg (FR), Dallas (DA), Chicago (CG), Stuttgart (ST), Turin (TO) and New York (NY) with extraction times from GTFS. ${\cal S}$ are the stations, $V$ the graph nodes, $E$ the graph edges and ${\cal L}$ the transit lines. $M$ is the maximum number of lines per edge.\label{TBL:datasets}}
  \vspace{-3mm}
  \centering
	\footnotesize
	{\renewcommand{\baselinestretch}{1.13}\normalsize
	\setlength\tabcolsep{3pt}
	\begin{tabular*}{.475\textwidth}{@{\extracolsep{\fill}} l r r r r r r r r r r r}
							& & & \multicolumn{4}{ c }{\footnotesize Line graph} & & \multicolumn{4}{ c }{\footnotesize Core  graph} \\
							\cline{4-7} \cline{9-12} \\[-2ex] \hline\noalign{\smallskip}
							& $t_{\text{extr}}$ & $|{\cal S}|$ & $|V|$ & $|E|$ & $|{\cal L}|$ & $M$ & & $|V|$ & $|E|$ & $|{\cal L}|$ & $M$ \Hhline
		FR	  & 0.7\Hs	& 74	&  80	&  81 &  5 & 4 &	& 20	& 21	&   5 & 4	\\
		DA & 3\Hs & 108	& 117	& 118	&  7 & 4 &	& 24	& 24	&   7 & 4	\\
		CG	& 13.5\Hs	& 143	& 153	& 154	&  8 & 6 & 	& 23	& 24	&   8 & 6	\\
		ST	  & 7.7\Hs	& 192	& 223	& 235	& 15 & 8 & 	& 51	& 60	&  15 & 8	\\
		TO	      & 4.9\Hs & 339	& 398 & 435	& 14 & 5 & 	& 91	& 124	&  14 & 5	\\
		NY
		                    & 3.7\Hs & 456	&  517	& 548	& 26 & 9 &	& 110	& 138	&  23 & 9	\Hhline
	\end{tabular*}}

\end{table}

For each dataset, we considered two versions of the line graph: the baseline graph and the core graph.
For each graph, we considered three ILP variants: the baseline ILP (B), the improved ILP (I) and the improved ILP with added separation penalty (S).
For each ILP, we evaluated two solvers: the GNU Linear Programming Kit (GLPK) and the COIN-OR CBC solver.
As most of the datasets (except Turin) still only had one connected component after applying the splitting rules described in Sect.~\ref{SEC:coreprobgraph}, we did not evaluate their application.
Tests were run on an Intel Core i5-6300U machine with 4 cores \`{a} 2.4 GHz and 12 GB RAM.
The CBC solver was compiled with multithreading support, and used with the default parameters and \texttt{threads=4}.
The GLPK solver was used with the feasibility pump heuristic (\texttt{fp\_heur=ON}), the proximity search heuristic (\texttt{ps\_heur=ON}) and the presolver enabled (\texttt{presolve=ON}).

For each node $v$, the penalty for a crossing between edge pairs ($\{A, B\}$ in Figure~\ref{FIG:crossings}, left) was $4 \cdot \degv$, for other crossings ($\{A, B\}$ in Figure~\ref{FIG:crossings}, right) it was $\degv$. The line separation penalty was $3\cdot\degv$. We found that these penalties produced nicer maps than a uniform penalty. This would imply $w_{\cal S\times} = 4 \cdot \max_{v\in V} \degv$ and $w_{\cal S\|} = 3 \cdot \max_{v\in V} \degv$. However, we found that moving some crossings or separations to stations with a degree greater than $2$ yielded results that looked better. Hence, crossings in $v \in \mathcal{S}$ were punished with $w_{\cal S\times}$ if $\degv = 2$ and otherwise with $3\cdot\degv$ (normal crossing) or $12\cdot\degv$ (edge-pair crossing). Similarly, in-station line separations where punished with $w_{\cal S\|}$ if $\degv = 2$ and $3\cdot\degv$ otherwise. Note that Lemma~\ref{LEM:crossingmoveing} still holds because we did not change the punishment for degree 2 stations. Also note that separations were only considered in $(S)$ and thus depended on the solver and the input order in $(B)$ and $(I)$.

\def\Hdimh{\footnotesize rows\hspace{0.3mm}{\footnotesize$\times$}\hspace{0.3mm}cols}
\def\Hdim#1#2{#1\hspace{0.1mm}{\footnotesize$\times$}\hspace{0.1mm}#2}
\def\Htglpk{\footnotesize GLPK}
\def\Htcbc{\footnotesize CBC}
\def\Hlong{---\phantom{\Hs}}
\def\Hno{---}
\def\Hgr{$>$}
\renewcommand*{\thefootnote}{\fnsymbol{footnote}}
\begin{table}
  \caption[]{Dimensions and solution times for Chicago (CG), Stuttgart (ST), Turin (TO) and New York (NY) and our three ILPs: baseline (B), improved (I), and with line separation penalty (S), with or without reduction to the core graph. A time of --- means we aborted after 12 hours. The last two columns show the number of crossings ($\times$) and separations ($||$) after optimization.\label{TBL:evalres}}
  \vspace{-3mm}
	\centering
	\footnotesize
	{\renewcommand{\baselinestretch}{1.13}\normalsize
		\setlength\tabcolsep{2pt}
	\begin{tabular*}{.475\textwidth}{@{\extracolsep{\fill}} l@{\hskip 1.2mm} c r r@{\hskip 2.5mm} r r r r@{\hskip 1.5mm}r@{\hskip 1mm}r r r}
							&& \multicolumn{3}{c}{\footnotesize On baseline graph} & & \multicolumn{3}{c}{\footnotesize On core graph} \\
							\cline{3-5} \cline{7-9} \\[-2ex] \hline\noalign{\smallskip}
							&& \Hdimh & \Htglpk & \Htcbc &  & \Hdimh & \Htglpk & \Htcbc & $\times$ & $||$ \Hhline


		CG   & B & \Hdim{41k}{861}   & \Hlong &  \Hlong & &  \Hdim{8.2k}{266} &  \Hlong & 47\Hm &   22 &   4-7 \\
              & I & \Hdim{1.7k}{1.2k} &  18\Hs & 0.4\Hs & &   \Hdim{484}{352} & 0.4\Hs &  0.2\Hs &   22 &   4-7 \\
		          & S & \Hdim{2.2k}{1.4k} &  47\Hm &   25\Hs & &   \Hdim{595}{405} &   26\Hs &  3.8\Hs &   27 &     0\Hhline

		ST & B	& \Hdim{224k}{2.5k} & \Hlong &  \Hlong & &   \Hdim{44k}{999} &  \Hlong &  \Hlong & \Hno &  \Hno \\
        		  & I & \Hdim{5.1k}{3.4k} & \Hlong &  7.1\Hs & & \Hdim{1.9k}{1.4k} &   10\Hs &  1.4\Hs &   65 & 8-15 \\
		          & S	& \Hdim{6.6k}{4.1k}	& \Hlong &  4.5\Hm & & \Hdim{2.5k}{1.6k} &  \Hlong &   36\Hs &   69 &     3 \Hhline

		TO     & B & \Hdim{24k}{2.1k}  & \Hlong &   \Hlong & &   \Hdim{13k}{1k}     &  \Hlong  &  \Hlong &\Hno& \Hno   \\
		          & I & \Hdim{4k}{2.8k}   & 32\Hm & 0.6\Hs   & &   \Hdim{1.9k}{1.4k}  &  7.1\Hs &  0.3\Hs &    78     &   6-10\\
		          & S & \Hdim{5k}{3.3k}   & \Hlong  &  9.1\Hs   & &   \Hdim{2.4k}{1.6k} &  \Hlong   &  4.6\Hs &    81     &     2 \Hhline

		NY  & B	& \Hdim{229k}{5.2k} & \Hlong &  \Hlong & &  \Hdim{96k}{2.3k} &  \Hlong &  \Hlong & \Hno &  \Hno \\
        		  & I &  \Hdim{11k}{7.1k} & \Hlong &  2.2\Hs & & \Hdim{4.5k}{3.1k} &   \Hlong &  0.7\Hs &  127 &     6-14 \\
		          & S	&  \Hdim{14k}{8.5k} & \Hlong &  2.4\Hm & & \Hdim{5.7k}{3.7k} &  \Hlong &   55\Hs &  132	&     2 \Hhline
	\end{tabular*}}
\end{table}

Table~\ref{TBL:evalres} shows the results of the ordering optimizations for 4 of our 6 datasets. With (I), the optimal order could be found in under 1.5 seconds for each of those datasets, and in under 1 minute with (S). For medium-sized networks an optimal solution for (S) was usually found in under $5$ seconds. Although the ILPs for (S) were only slightly larger than for (I), optimization on the core graph took $35$ times longer on average (with CBC). (B) could not be optimized in under $12$ hours for all datasets except Chicago. In general, CBC outperformed GLPK for larger datasets, sometimes dramatically. As expected, core graph reduction made the ILPs significantly smaller. On average, the number of rows decreased by $61\,\%$ and the number of columns by $59\,\%$ for (I). For (S), the decrease was $62\,\%$ and $60\,\%$, respectively.

%

\section{Conclusions and Future Work}\label{SEC:conclusions}

We have evaluated LOOM and shown that it produces geographically accurate transit maps. The whole pipeline took less than 1 minute for all considered inputs (including the rendering step).
As the line graph construction required more time than the subsequent ILP solution for some datasets, faster algorithms for extracting the line graph would be of interest.

The ideas behind LOOM may be useful also in a non-transit scenario.
For example, one closely related problem is that of wire routing in integrated-circuit design.
There, stations correspond to chips and other elements (which in wire routing are indeed of polygonal form), lines correspond to wires, and the geographical course of the lines may correspond to a pre-existing wiring.

%
%
%

\bibliographystyle{ACM-Reference-Format}
\bibliography{transitmaps}

\end{document}